\pdfoutput=1

\documentclass{imsart}

\usepackage{amsfonts}
\usepackage{amsmath}
\usepackage{amsthm}
\usepackage{amssymb}
\usepackage{dsfont}
\usepackage{epsfig}
\usepackage{apacite}

\RequirePackage[OT1]{fontenc}
\RequirePackage{amsthm,amsmath,natbib,verbatim}
\RequirePackage[colorlinks=true,citecolor=blue,linkcolor=blue,urlcolor=blue]{hyperref}

\startlocaldefs
\theoremstyle{plain}
\newtheorem{thm}{Theorem}
\newtheorem{lem}[thm]{Lemma}
\newcommand{\ind}{{\mathds{1}}}
\DeclareMathOperator{\Prob}{Pr}
\DeclareMathOperator{\Exp}{E}
\newcommand{\bs}{\boldsymbol}
\newcommand{\eqd}{\triangleq}
\newcommand{\myvec}{\bs}
\endlocaldefs

\begin{document}

\begin{frontmatter}
\title{Conservative Hypothesis Tests and Confidence Intervals using Importance Sampling}
\runtitle{Conservative Hypothesis Tests and Confidence Intervals}

\begin{aug}
\author{Matthew T. Harrison}\thanksref{t1}

\address{Division of Applied Mathematics \\ Brown University \\ Providence, RI 02912 \\
email: {\tt Firstname\_Lastname@Brown.edu}}

\thankstext{t1}{This work was supported in part by the National Science Foundation (NSF) grant DMS-1007593 and also, while the author was in the Department of Statistics at Carnegie Mellon University, by the NSF grant DMS-0240019 and the National Institutes of Health grant NIMH-2RO1MH064537.  The author thanks Stuart Geman, Jeffrey Miller, and Lee Newberg for helpful comments.  Shigeyoshi Fujisawa and Gy\"orgy Buzs{\'a}ki graciously provided the data for the example in Section \ref{s:jitter}.}
\runauthor{M.T. Harrison}

\affiliation{Brown University}

\end{aug}

\begin{abstract} {} Importance sampling is a common technique for Monte Carlo approximation, including Monte Carlo approximation of p-values.  Here it is shown that a simple correction of the usual importance sampling p-values creates {\em valid} p-values, meaning that a hypothesis test created by rejecting the null when the p-value is $\leq \alpha$ will also have a type I error rate $\leq \alpha$.  This correction uses the importance weight of the original observation, which gives valuable diagnostic information under the null hypothesis.  Using the corrected  p-values can be crucial for multiple testing and also in problems where evaluating the accuracy of importance sampling approximations is difficult.  Inverting the corrected  p-values provides a useful way to create Monte Carlo confidence intervals that maintain the nominal significance level and use only a single Monte Carlo sample.  Several applications are described, including accelerated multiple testing for a large neurophysiological dataset and exact conditional inference for a logistic regression model with nuisance parameters.
\end{abstract}



\tableofcontents
\end{frontmatter}

\section{Introduction}


Importance sampling is a common technique for Monte Carlo approximation, including Monte Carlo approximation of p-values.  Besides its use in situations where efficient direct-sampling algorithms are unavailable, importance sampling can be used to accelerate the approximation of tiny p-values as needed for multiple hypothesis testing.  Importance sampling can also be used for Monte Carlo approximation of confidence intervals using a single Monte Carlo sample by inverting a family of hypothesis tests.  The practicality of these two uses of importance sampling --- accelerated multiple testing and Monte Carlo approximation of confidence intervals --- would seem to be limited by the fact that each places strong requirements on the importance sampling procedure.  Multiple testing controls can be sensitive to tiny absolute errors (but large relative errors) in small p-values, which would seem to demand either excessively large Monte Carlo sample sizes or unrealistically accurate importance sampling proposal distributions in order to reduce absolute Monte Carlo approximation errors to tolerable levels.  The confidence interval procedure uses a single proposal distribution to approximate probabilities under a large family of target distributions, which again would seem to demand large sample sizes in order to overcome the high importance sampling variability that is frequently encountered when the proposal distribution is not tailored to a specific target distribution.  Nevertheless, as shown here, simple corrections of the usual importance sampling p-value approximations can be used to overcome these difficulties, making importance sampling a practical choice for accelerated multiple testing and for constructing Monte Carlo confidence intervals.  

The p-value corrections that we introduce require negligible additional computation and still converge to the target p-value, but have the additional property that they are valid p-values, meaning the probability of rejecting a null hypothesis is $\leq\alpha$ for any specified level $\alpha$.  Hypothesis tests and confidence intervals constructed from the corrected p-value approximations are guaranteed to be conservative, regardless of the Monte Carlo sample size, while still behaving like the target tests and intervals for sufficiently large Monte Carlo sample sizes.  The combination of being both conservative and consistent turns out to be crucial in many applications where the importance sampling variability cannot be adequately controlled with practical amounts of computing, including multiple testing, confidence intervals, and any hypothesis testing situation where the true variability of the importance sampling algorithm is either large or unknown.  We demonstrate the practical utility of the correction with several examples below.     

Let $X$ denote the original observation (data set).  Assume that the null hypothesis specifies a known distribution $P$ for $X$.  For a specified test statistic $t$, the goal is to compute the p-value $p(X)$ defined by
\[ p(x) \eqd \Prob\left(t(X) \geq t(x)\right) \]
where the probability is computed under the null hypothesis.  Importance sampling can be used to Monte Carlo approximate $p(X)$ if the p-value cannot be determined analytically.  Let $\myvec{Y}\eqd(Y_1,\dotsc,Y_n)$ be an independent and identically distributed (i.i.d.)~sample from a distribution $Q$ (the {\em proposal distribution}) whose support includes the support of $P$ (the {\em target distribution}).  Then $p(X)$ can be approximated with either
\[ \begin{aligned}  \widehat p(X,\myvec{Y}) & \eqd \frac{\sum_{i=1}^n w(Y_i)\ind\bigl\{t(Y_i)\geq t(X)\bigr\}}{n} \\ 
\widetilde p(X,\myvec{Y}) & \eqd \frac{\sum_{i=1}^n w(Y_i)\ind\bigl\{t(Y_i)\geq t(X)\bigr\}}{\sum_{j=1}^n w(Y_j)} \end{aligned} \]
where $\ind$ is the indicator function and where the {\em importance weights} are defined to be
\[ w(x) \eqd \frac{P(x)}{Q(x)} \]
in the discrete case, and the ratio of densities in the continuous case.  Each of these are consistent approximations of $p(X)$ as the Monte Carlo sample size increases.  The former is unbiased and is especially useful for approximating extremely small p-values.  The latter can be evaluated even if the importance weights are only known up to a constant of proportionality.  Note that $n$ is the Monte Carlo sample size; the sample size or dimensionality of $X$ is irrelevant for the developments here.  The reader is referred to \citet{Liu:Monte:2001} for details and references concerning importance sampling and to \citet{Lehmann:Testing:2005} for hypothesis testing.

Here we propose the following simple corrections of $\widehat p$ and $\widetilde p$ that make use of the importance weight of the original observation, namely,
\[ \begin{aligned}  \widehat p_*(X,\myvec{Y}) & \eqd \frac{w(X)+\sum_{i=1}^n w(Y_i)\ind\bigl\{t(Y_i)\geq t(X)\bigr\}}{1+n} \\ 
\widetilde p_*(X,\myvec{Y}) & \eqd \frac{w(X)+\sum_{i=1}^n w(Y_i)\ind\bigl\{t(Y_i)\geq t(X)\bigr\}}{w(X)+\sum_{j=1}^n w(Y_j)} \end{aligned} \]
We will show that the corrected p-value approximations, while clearly still consistent approximations of the target p-value $p(X)$, are also themselves {\em valid} p-values, meaning 
\[ \Prob\bigl(\widehat p_*(X,\myvec{Y}) \leq \alpha\bigr)\leq \alpha \quad \quad \text{and} \quad \quad \Prob\bigl(\widetilde p_*(X,\myvec{Y}) \leq \alpha\bigr)\leq \alpha \]
for all $\alpha\in[0,1]$ and $n\geq 0$ under the null hypothesis, where the probability is with respect to the joint distribution of data and Monte Carlo sample.  
These simple corrections have far-reaching consequences and enable importance sampling to be successfully used in a variety of situations where it would otherwise fail.  Of special notes are the ability to properly control for multiple hypothesis tests and the ability to create valid confidence intervals using a single Monte Carlo sample.  


\section{Main results} \label{s:m}
    
The main results are that $\widehat p_*$ and $\widetilde p_*$ are valid p-values (Theorems \ref{t:hat} and \ref{t:tilde}). 
We generalize the introductory discussion in two directions.  First, we allow arbitrary distributions, so the importance weights become the Radon-Nikodym derivative $dP/dQ$.  In the discrete case this simplifies to the ratio of probability mass functions, as in the introduction, and in the continuous case this simplifies to the ratio of probability density functions.  Second, we allow the choice of test statistic to depend on $(X,Y_1,\dotsc,Y_n)$ as long as the choice is invariant to permutations of $(X,Y_1,\dotsc,Y_n)$.  We express this mathematically by writing the test statistic, $t$, as a function of two arguments: the first is the same as before, but the second argument takes the entire sequence $(X,Y_1,\dotsc,Y_n)$, although we require that $t$ is invariant to permutations in the second argument.  For example, we may want to transform the sequence $(X,Y_1,\dotsc,Y_n)$ in some way, either before or after applying a test-statistic to the individual entries.  As long as the transformation procedure is permutation invariant (such as centering and scaling, or converting to ranks), everything is fine.  Transformations are often desirable in multiple testing contexts for improving balance \citep{westfall1993resampling}.

We begin with the precise notation and assumptions for the theorems.  $P$ and $Q$ are probability distributions defined on the same measurable space $(S,\mathcal{S})$ with $P\ll Q$ (meaning that sets with positive $P$ probability also have positive $Q$ probability), and $w$ is a fixed, nonnegative version of the Radon-Nikodym derivative $dP/dQ$.  $\mathcal{M}$ denotes the set of all $(n+1)!$ permutations $\myvec{\pi} \eqd (\pi_0,\dotsc,\pi_n)$ of $(0,\dotsc,n)$.  For $\myvec{\pi}\in\mathcal{M}$ and $\myvec{z} \eqd (z_0,\dotsc,z_n)$, we define $\myvec{z^{(\myvec{\pi})}} \eqd (z_{\pi_0},\dotsc,z_{\pi_n})$.
Assume that $t:S\times S^{n+1}\mapsto[-\infty,\infty]$ has the property that
\[ t(a,\myvec{z})  = t(a,\myvec{z^{(\myvec{\pi})}}) \]
for all $a\in S$, $\myvec{z}\in S^{n+1}$, and $\myvec{\pi}\in\mathcal{M}$.

For $\myvec{z}\in S^{n+1}$ define
\begin{align*} 
\widehat p_*(\myvec{z}) & \eqd  \frac{\sum_{i=0}^n w(z_i)\ind\{t(z_i,\myvec{z}) \geq t(z_0,\myvec{z})\}}{n+1} \\ 
\widetilde p_*(\myvec{z}) & \eqd    \frac{\sum_{i=0}^n w(z_i)\ind\{t(z_i,\myvec{z}) \geq t(z_0,\myvec{z})\}}{\sum_{j=0}^n w(z_j)} \end{align*}
where we take $0/0 \eqd 0$.  Let $X$ have distribution $P$ and let $Y_0,Y_1,\dotsc,Y_n$ be an i.i.d.~sample from $Q$, independent of $X$.  For notational convenience define $\myvec{Z} \eqd (Z_0,Z_1,\dotsc,Z_n)$ by
\[ Z_0 \eqd  X, \ Z_1 \eqd Y_1 , \ \dotsc , \ Z_n \eqd Y_n  \]
so that the corrected p-values are $\widehat p_*(\myvec{Z})$ and $\widetilde p_*(\myvec{Z})$.  Then,  

\begin{thm} \label{t:hat} $\Prob\bigl(\widehat p_*(\myvec{Z}) \leq \alpha\bigr) \leq \alpha$ for all $\alpha\in[0,1]$.
\end{thm}

\begin{thm} \label{t:tilde} $\Prob\bigl(\widetilde p_*(\myvec{Z}) \leq \alpha\bigr) \leq \alpha$ for all $\alpha\in[0,1]$.
\end{thm}

Proofs are in the Appendix.  
The theorems do not require any special relationships among $P$, $Q$, $t$, or $n$.  For example, in parametric settings $Q$ does not need to be in the same model class as the null and/or alternative.  Validity of the corrected p-values is ensured even for unusual cases such as $n=0$ or importance weights with infinite variance.  We discuss some practical considerations for choosing $Q$ in the next section (presumably, $P$, $t$ and $n$ are dictated by the application and the computational budget).  Validity of the corrected p-values is well known for the special case of direct sampling, i.e., $Q\equiv P$ and $w\equiv 1$.  For Markov chain Monte Carlo (MCMC) approaches, \citet{besag1989generalized} demonstrate how to generate valid p-value approximations using techniques that are unrelated to the ones here.  

The theorems continue to hold if each $Y_k$ is chosen from a different $Q$, say $Q_k$, as long as the sequence $Q_0,Q_1,\dotsc,Q_n$ is itself i.i.d.~from some distribution over proposal distributions.  The $k$th importance weight is now a fixed version of the Radon-Nikodym derivative $dP/dQ_k$.  This generalization can be useful in practice when each $Y_k$ is generated hierarchically by first choosing an auxiliary variable $V_k$ and then, given $V_k$, choosing $Y_k$ from some distribution $Q^{V_k}$ that depends on $V_k$.  If the marginal distribution (i.e., $Q$) of $Y_k$ is not available, one can use importance weights based on the conditional distribution (i.e., $Q^{V_k}$) of $Y_k$ given $V_k$.  This is equivalent to using random proposal distributions as described above.  The drawback of combining this approach with the p-value corrections, is that the importance weight of the original observation is evaluated using $dP/dQ_0$, where $Q_0$ is a randomly chosen proposal distribution.  This further increases the Monte Carlo randomness already inherent in the p-value approximations.  Note that generalizing the theorems to allow for random proposals requires no additional work.  If $\nu$ is the joint distribution of each $(Q_k,Y_k)$ and $\nu_Q$ is the marginal distribution of each $Q_k$, simply apply the theorems with $\nu_Q\times P$ in place of $P$, $\nu$ in place of $Q$, $(Q_0,X)$ in place of $X$, and $(Q_k,Y_k)$ in place of $Y_k$.  If one uses a regular conditional distribution to define the importance weight $[d(\nu_Q\times P)/d\nu](Q_k,x)$, then it will simplify to a version of $[dQ_k/dP](x)$.

Finally, we recall the well known fact that any valid family of p-values can be inverted in the usual way to give valid confidence intervals.  In particular, consider a collection of distributions $\{P_\theta:\theta\in\Theta\}$ and, for each $\theta\in\Theta$, let $p_*(\theta,\myvec{Z})$ be any valid p-value for testing the null hypothesis of $P_\theta$.  Fix $\alpha\in[0,1]$ and define the random set 
\[ C^\alpha(\myvec{Z}) \eqd \bigl\{\theta\in\Theta:p_*(\theta,\myvec{Z})>\alpha\} \]
Then $C^\alpha(\myvec{Z})$ is a $1-\alpha$ confidence set for $\theta$ because
\[ \Prob_\theta\bigl(C^\alpha(\myvec{Z})\ni\theta\bigr) = 1 - \Prob_\theta\bigl(p_*(\theta,\myvec{Z})\leq\alpha\bigl) \geq 1-\alpha \]
where $\Prob_\theta$ is calculated under the null hypothesis of $P_\theta$.  An appealing feature about inverting either $\widehat p_*$ or $\widetilde p_*$ is that the same importance samples can (in principle) be used for testing each $\theta$.  Only the importance weights (and perhaps the test statistic) vary with $\theta$.  Section \ref{s:rasch} below illustrates this application.  The idea of using importance sampling to construct confidence intervals from a single Monte Carlo sample was pointed out in \citet{PJGgeyer1992constrained}
.  See \citet{bolviken1996confidence} and \citet{garthwaite1992generating} for examples of other ways to create Monte Carlo confidence intervals.

\section{Practical considerations} \label{s:power}

It is clear that each of the corrected p-values has the same asymptotic behavior as its uncorrected counterpart, so for sufficiently large $n$ they are essentially equivalent.  But in practice, $n$ will often be too small.  The concern for the investigator is that the corrected p-values may behave much worse than the uncorrected ones for practical choices of $n$.  In particular, the investigator is concerned about the following situation: the null hypothesis is false; the target p-value, $p$, is small; the uncorrected approximations, $\widehat p$ or $\widetilde p$, are also small; but the corrected approximations, $\widehat p_*$ or $\widetilde p_*$, are large.  Careful choice of $Q$ can lower the chances of this situation.

Each of the corrected p-values can be expressed as an interpolation between their respective uncorrected versions and another (non-Monte Carlo) valid p-value:
\begin{gather}
\widehat p_*(\myvec{z}) = \frac{1}{n+1}w(x) + \left(1-\frac{1}{n+1}\right)\widehat p(\myvec{z}) \label{e:phat2} \\
\widetilde p_*(\myvec{z}) = \frac{w(x)}{w(x)+\sum_{j=1}^n w(y_j)} + \left(1-\frac{w(x)}{w(x)+\sum_{j=1}^n w(y_j)}\right)\widetilde p(\myvec{z}) \label{e:ptilde2}
\end{gather}
where we are using the shorthand $\myvec{z} \eqd (x,y_1,\dotsc,y_n)$.    In both cases, power suffers when $X$ comes from a distribution in the alternative hypothesis, but $w(X)$ is typically large relative to $w(Y_1),\dotsc,w(Y_n)$.  Since $w(Y_k)$ always has mean 1, problems might arise for alternatives that tend to make $w(X)$ much larger than 1.  This problem can be avoided by choosing $Q$ so that it gives more weight than does $P$ to regions of the sample space that are more typical under alternatives than they are under $P$.  The problem can also be avoided by choosing $Q$ to be similar to $P$ so that the weights are close to 1 throughout the sample space.  Most proposal distributions are designed with one of these two goals in mind, so the corrected p-values should behave well for well-designed proposal distributions.
In practice, however, proposal distributions can be quite bad, and it can be helpful to look more closely at how the proposal affects the power of the corrected p-values.  

From \eqref{e:phat2} we see that $\widehat p_*$ is an interpolation between $\widehat p$ and $w(x)$, the latter of which is a valid p-value.  Validity of $w(X)$ can be seen either by Theorem \ref{t:hat} for $n=0$ or by the simple calculation
\begin{equation} \label{e:NP} \Prob(w(X) \leq \alpha) = \sum_x \ind\bigl\{P(x)\leq\alpha Q(x)\bigr\} P(x) \leq \sum_x \alpha Q(x) = \alpha \end{equation}
which easily generalizes to more abstract settings.
Testing $w(X)=P(X)/Q(X)\leq\alpha$ is simply a likelihood ratio test (LRT) of $P$ versus $Q$, although using the critical value of $\alpha$ will give a test of size smaller than $\alpha$.  The test can be strongly conservative, because the bound in \eqref{e:NP} can be far from tight.  So $\widehat p_*$ is an interpolation between a conservative LRT of $P$ versus $Q$ and an (often liberal) importance sampling Monte Carlo approximation of the target p-value.  If the effect of $w(X)$ on $\widehat p_*$ has not disappeared (such as for small $n$), then $\widehat p_*$ is a p-value for the original null, $P$, versus a new alternative that is some amalgamation of the original alternative and the proposal distribution $Q$.  If $Q$ is not in the alternative (as it often will not be), this modified alternative is important to keep in mind when interpreting any rejections.

The effect of $Q$ on interpreting rejections is especially important in multiple-testing situations.  In these settings, importance sampling is often used to accelerate the approximation of tiny p-values via $\widehat p$.  If $\widehat p\approx 0$, however, then $\widehat p_*\approx w(x)/(n+1)$ and $Q$ plays a critical role in determining which null hypotheses are rejected after correcting for multiple comparisons.  The investigator can take advantage of this by ensuring that a rejection of $P$ in favor of $Q$ is sensible for the problems at hand, and by ensuring that events with small p-values will be heavily weighted by $Q$.  

Turning to \eqref{e:ptilde2}, we see that $\widetilde p_*$ is an interpolation between $1$ and $\widetilde p$.  Since $\widetilde p \leq 1$, we always have
\[ \widetilde p_* \geq \widetilde p \]
so that $\widetilde p_*$ always leads to a more conservative test than $\widetilde p$.  The degree of conservativeness is controlled by the ratio $w(x)/\sum_j w(y_j)$.  If the ratio is small, then the correction ensures validity with little loss of power.  If it is large, there may be a substantial loss of power when using the correction.  The ratio will be approximately $1/n$ if $P$ and $Q$ are similar, which is often the design criteria when using $\widetilde p$.

An important caveat for the main results is that $Q$ is {\em not} allowed to depend on the observed value of $X$.  This caveat precludes one of the classical uses of importance sampling, namely, using the observed value of $t(X)$ to design a $Q$ that heavily weights the event $\{x:t(x)\geq t(X)\}$.  In many cases, however, since the functional form of $t$ is known, one can {\em a priori} design a $Q$ that will heavily weight the event $\{x:t(x)\geq t(X)\}$ whenever the p-value would be small.  For example, given a family of proposal distributions $\{Q_\ell\}_\ell$, each of which might be useful for a limited range of observed values of $t(X)$, one can use finite mixture distributions of the form $Q=\sum_{\ell=1}^L \lambda_\ell Q_\ell$ to obtain more robust performance.  For similar reasons, mixture distributions are also useful for creating Monte Carlo confidence intervals in which the same proposal distribution is used for a family of target distributions.  The examples in Sections \ref{s:rasch} and \ref{s:jitter} have more details, and \citet{hesterberg1995weighted} contains a more general discussion of the utility of mixture distributions for importance sampling.  Finally, we note that this caveat about $Q$ does not preclude conditional inference.  In conditional inference $P$ is the null conditional distribution of $X$ given some appropriate $A(X)$, and all of the analysis takes place after conditioning on $A(X)$.  The choice of $Q$ (and also $t$, for that matter) can thus depend on $A(X)$, but not on additional details about $X$.       

\section{Applications} \label{s:ex}

\subsection{Accelerating multiple permutation tests} \label{s:ex1}

Consider a collection of $N$ datasets.  The $i$th dataset $X^i \eqd (\myvec{V}^i,\myvec{L}^i)$ contains a sample of values $\myvec{V}^i \eqd (V^i_1,\dots,V^i_{m_i})$ and corresponding labels $\myvec{L}^i \eqd (L^i_1,\dotsc,L^i_{m_i})$.  The distributions over values and labels are unknown and perhaps unrelated across datasets.  We are interested in identifying which datasets show a dependence between the values and the labels.  From a multiple hypothesis testing perspective, the $i$th null hypothesis is that $\myvec{V}^i$ and $\myvec{L}^i$ are independent, or more generally, that the values of $\myvec{L}^i$ are exchangeable conditioned on the values of $\myvec{V}^i$.  Given a test statistic $t(x) \eqd t(\myvec{v},\myvec{\ell})$, a permutation test p-value for the $i$th null hypothesis is given by 
\[ p(X^i)  \eqd  \frac{1}{m_i!}\sum_{\myvec{\pi}} \ind\bigl\{t(\myvec{V}^i,(\myvec{L}^i)^{(\myvec{\pi})})\geq t(\myvec{V}^i,\myvec{L}^i)\bigr\}  \]
where the sum is over all permutations $\myvec{\pi} \eqd (\pi_1,\dotsc,\pi_{m_i})$ of $(1,\dotsc,m_i)$, and where the notation $\myvec{\ell^{(\myvec{\pi})}} \eqd (\ell_{\pi_1},\dotsc,\ell_{\pi_{m_i}})$ denotes a permuted version of the elements of $\myvec{\ell}$ using the permutation $\myvec{\pi}$.  The beauty of the permutation test is that it converts a large composite null (independence of values and labels) into a simple null (all permutations of the labels are equally likely) by conditioning on the values and the labels but not their pairings.  The $i$th null distribution, $P^i$, becomes the uniform distribution over permutations (of the labels).  The Bonferroni correction can be used to control the family-wise error rate (FWER), i.e., the probability of even a single false rejection among all $N$ datasets \citep{Lehmann:Testing:2005}.  In particular, rejecting null hypothesis $i$ whenever $p(X^i)\leq\alpha/N$ ensures that the probability of even a single false rejection is no more than $\alpha$.  Although computing $p$ exactly is often prohibitive, Monte Carlo samples from $P^i$ are readily available for approximating $p$, and importance sampling can be used to accelerate the approximation of tiny p-values.  Bonferroni is still sensible as long as the Monte Carlo approximate p-values are valid p-values.
 
Here is a specific simulation example.  Consider $N=10^4$ independent data\-sets, each with $m_i=m=100$ real-valued values and corresponding binary labels.  In each case there are 40 labels of 1 and 60 labels of 0.  In $9990$ datasets, the labels and values are independent, in which case the values are i.i.d.~standard Cauchy (i.e., with scale parameter 1).  In 10 of the datasets, the labels and values are dependent, in which case the values associated to label 0 are i.i.d.~standard Cauchy and the values associated to label 1 are i.i.d.~standard Cauchy plus 2 (i.e., a standard Cauchy shifted to the right by 2).  (This particular example was chosen because standard tests like the the two-sample $t$-test and the Wilcoxon rank-sum test tend to perform poorly.)  For the alternative hypothesis that label 1 tends to have larger values than label 0, a sensible test statistic is the difference in medians between values associated to label 1 and values associated to label 0, namely,
\[ t(x)  \eqd  t(\myvec{v},\myvec{\ell})  \eqd  \text{median}(v_j:\ell_j=1)-\text{median}(v_j:\ell_j=0)  \] 

Table \ref{t:intro} shows the Bonferroni-corrected performance of several different Monte Carlo approaches as the Monte Carlo sample size ($n$) increases.  The approximations $\bar p$ and $\bar p_*$ refer to $\widehat p$ and $\widehat p_*$, respectively, for the special case of $Q\equiv P$ and $w\equiv 1$, i.e., direct sampling (all permutations are equally likely).  Using either $\bar p$ or $\bar p_*$ would be the standard method of approximating $p$.  The approximations $\widehat p$ and $\widehat p_*$ are based on a proposal distribution that prefers permutations that pair large values of $t$ with label $1$ (see the Appendix for details).  In each case we reject when the approximate p-value is $\leq \alpha/N$.  The final approximation $\widehat q$ refers to a Wald $1-\alpha/(2N)$ upper confidence limit for $\widehat p$ using the same importance samples to estimate a standard error.  For $\widehat q$ we reject whenever $\widehat q \leq \alpha/(2N)$, which, if the confidence limits were exact, would correctly control the FWER at level $\alpha$.   

\begin{table}[h]
\caption{Bonferroni-corrected testing performance for different p-value approximations versus Monte Carlo sample size ($n$); $H_0$ is false for $10$ out of $N=10^4$ tests}
\label{t:intro}
\centering
\begin{tabular}{ c | c c c c || c c c c} 
\multicolumn{1}{c}{ } & \multicolumn{4}{c}{\# correct rejections} & \multicolumn{4}{c}{\# incorrect rejections} \\
$n$ & $10^1$ & $10^2$ & $10^3$ & $10^4$ & $10^1$ & $10^2$ & $10^3$ & $10^4$ \\ \hline
$\bar p$ & 10 & 10 & 10 & 10 & 908 & 87 & 12 & 1 \\ 
$\bar p_*$ & 0 & 0 & 0 & 0 & 0 & 0 & 0 & 0 \\ \hline
$\widehat p$ & 9 & 8 & 7 & 5 & 7579 & 1903 & 2 & 0 \\
$\widehat p_*$ & 7 & 6 & 6 & 5 & 0 & 0 & 0 & 0 \\ \hline
$\widehat q$ & 6 & 5 & 4 & 3 & 4545 & 585 & 0 & 0
\end{tabular}
\end{table}

Only $\widehat p_*$ works well.  The approximations that are not guaranteed to be valid ($\bar p$, $\widehat p$, $\widehat q$) require excessively large $n$ before the false detection rate drops to acceptable levels.  The cases $n=10^3$ and $n=10^4$ (even the largest of which was still too small for $\bar p$ to reach the Bonferonni target of zero false rejections) are computationally burdensome, and the situation only worsens as $N$ increases.  Furthermore, in a real problem the investigator has no way of determining that $n$ is large enough.  The confidence limit procedure similarly failed because the estimated standard errors are too variable.

The valid p-values ($\bar p_*$, $\widehat p_*$) ensure that the Bonferonni correction works regardless of $n$, but $\bar p_*\geq 1/(n+1)$, so it is not useful after a Bonferonni correction unless $n$ is extremely large.  The good performance of $\widehat p_*$ requires the combination of importance sampling, which allows tiny p-values to be approximated with small $n$, and the validity correction introduced in this paper, which allows multiple-testing adjustments to work properly. 

\subsection{Exact inference for covariate effects in Rasch models} \label{s:rasch}

Let $X=(X_{ij}:i=1,\dotsc,M; \ j=1,\dotsc,N)$ be a binary $M\times N$ matrix.  Consider the following logistic regression model for $X$. The $X_{ij}$ are independent Bernoulli$(p_{ij})$ where
\begin{equation} \label{e:rasch} \log\frac{p_{ij}}{1-p_{ij}}  \eqd  \kappa + \alpha_i + \beta_j + \theta v_{ij} \end{equation}
for unknown coefficients $\kappa$, $\myvec{\alpha}=(\alpha_1,\dotsc,\alpha_M)$, and $\myvec{\beta}=(\beta_1,\dotsc,\beta_N)$, and known covariates $\myvec{v}=(v_{ij}:i=1,\dotsc,M; \ j=1,\dotsc,N)$.  The special case $\theta=0$ is called the Rasch model and is commonly used to model the response of $M$ subjects to $N$ binary questions \citep{rasch1961general}.  In this example, we discuss inference about $\theta$ when $\kappa,\myvec{\alpha},\myvec{\beta}$ are treated as nuisance parameters.

Consider first the problem of testing the null hypothesis of $\theta=\theta'$ versus the alternative of $\theta\neq\theta'$.  Conditioning on the row and column sums removes the nuisance parameters, and the original composite null hypothesis reduces to the simple (conditional) null hypothesis of $X\sim P_{\theta'}$, where $P_{\theta'}$ is the conditional distribution of $X$ given the row and column sums for the model in \eqref{e:rasch} with $\kappa=\alpha_1=\dotsb=\alpha_M=\beta_1=\dotsb=\beta_N=0$ and $\theta=\theta'$.  A sensible test statistic is the minimal sufficient statistic for $\theta$, namely,
\[ t(x)  \eqd  \sum_{ij}x_{ij}v_{ij} \]
Since $t(X)$ has power for detecting $\theta > \theta'$ and $-t(X)$ has power for detecting $\theta < \theta'$, it is common to combine upper- and lower-tailed p-values into a single p-value, $p^\pm(\theta',X)$, defined by
\begin{align} p^+(\theta,x) &  \eqd  \Prob_\theta\bigl(t(X)\geq t(x)\bigr) \notag
\\ p^-(\theta,x) &  \eqd  \Prob_\theta\bigl(-t(X)\geq -t(x)\bigr) \notag \\
 p^\pm(\theta,x) &  \eqd  \min\bigl\{1,2\min\{p^+(\theta,x),p^-(\theta,x)\}\bigr\} \label{e:pm}
\end{align}
where $\Prob_\theta$ uses $X\sim P_\theta$.
There are no practical algorithms for computing $p^\pm(\theta',X)$ nor for direct sampling from $P_{\theta'}$.  The importance sampling algorithm in \citet{Chen:Sequential:2005} (designed for the case $\theta=0$) can be modified, however, to create a sensible and practical proposal distribution, say $Q_\theta$, for any $P_\theta$.  The details of these proposal distributions will be reported elsewhere.  $Q_\theta(x)$ can be evaluated, but $P_\theta(x)$ is known only up to a normalization constant, so we must use $\widetilde p$ and $\widetilde p_*$.  Here we compute two Monte Carlo p-value approximations, one for each of $t$ and $-t$, and combine them as in \eqref{e:pm}.  Validity of the underlying p-values ensures validity the combined p-value.

Confidence intervals for $\theta$ can be constructed in the usual way by inverting this family of tests.  In particular, a $1-\alpha$ confidence set for $\theta$ is the set of $\theta'$ for which we do not reject the null hypothesis of $\theta=\theta'$ at level $\alpha$.  An annoying feature of inverting Monte Carlo hypothesis tests is that in most cases a new sample is needed for each $\theta'$, which increases the computational burden and creates pathologically shaped confidence sets, because the Monte Carlo randomness varies across $\theta'$ \citep{bolviken1996confidence}.  Importance sampling suffers from the same problems, {\em unless} the same importance samples, and hence the same proposal distribution, are used for testing each $\theta'$.  

To create a common proposal distribution that might work well in practice, we use a mixture of proposals designed for specific $\theta$.  In particular, define
\[ Q  \eqd  \frac{1}{L}\sum_{\ell=1}^L Q_{\theta_\ell} \]
where $(\theta_1,\dotsc,\theta_L)$ are fixed and each $Q_{\theta_\ell}$ is a proposal distribution appropriate for $P_{\theta_\ell}$ as mentioned earlier.  Here we use $L=601$ and $(\theta_1,\dotsc,\theta_L)=(-6.00,-5.98,-5.96,\dotsc,6.00)$.  For any $\theta$, the importance weights are 
\[ w(\theta,x)  \eqd  \frac{P_\theta(x)}{Q(x)} = c_\theta\frac{\exp\bigl(\theta t(x)\bigr)}{\sum_\ell Q_{\theta_\ell}(x)} \]
for any binary matrix $x$ with the correct row and column sums, where $c_\theta$ is an unknown constant that is not needed for computing $\widetilde p$ and $\widetilde p_*$.  
We will use $\widetilde p^+(\theta,\myvec{z})$ and $\widetilde p^-(\theta,\myvec{z})$ to denote $\widetilde p(\myvec{z})$ when computed using the importance weights $w(\theta,\cdot)$ and the test statistics $+t$ and $-t$, respectively, where $\myvec{z} \eqd (x,y_1,\dotsc,y_n)$.  Similarly, define $\widetilde p_*^+(\theta,\myvec{z})$ and $\widetilde p_*^-(\theta,\myvec{z})$.  The upper and lower p-values can be combined via
\[ \widetilde p^\pm = \min\bigl\{1,2\min\{\widetilde p^+,\widetilde p^-\}\bigr\} \quad \text{and} \quad \widetilde p_*^\pm = \min\bigl\{1,2\min\{\widetilde p_*^+,\widetilde p_*^-\}\bigr\}  \]
In light of Theorem \ref{t:tilde}, $\widetilde p_*^\pm(\theta',\myvec{Z})$ is a valid p-value for testing the null hypothesis that $\theta=\theta'$.

Inverting these p-values gives the confidence sets
\begin{gather*} \widetilde C^\alpha(\myvec{z})  \eqd  \bigl\{\theta\in\mathbb{R}:\widetilde p^\pm(\theta,\myvec{z}) > \alpha\bigr\}  \\
\widetilde C_*^\alpha(\myvec{z})  \eqd  \bigl\{\theta\in\mathbb{R}:\widetilde p^\pm_*(\theta,\myvec{z}) > \alpha\bigr\}
\end{gather*}
For fixed $\myvec{z}$, the approximate p-values are well-behaved functions of $\theta$, so it is straightforward to numerically approximate the confidence sets (which will typically be intervals).

Table \ref{t:rasch} describes the results of a simulation experiment that investigated the coverage properties of $\widetilde C^\alpha$ and $\widetilde C^\alpha_*$.  In that experiment, we took $M=200$, $N=10$, and fixed the parameters $\kappa$, $\theta$, $\alpha$, $\beta$ and the covariates $v$ (see the Appendix for details).  We set $\theta=2$ and repeatedly (1000 times) generated datasets and importance samples in order to estimate the true coverage probabilities and the median interval lengths of various confidence intervals.  Confidence intervals based on the corrected p-values always maintain a coverage probability of at least $1-\alpha$ without becoming excessively long, while those based on uncorrected p-values can have much lower coverage probabilities.  As the number of importance samples increases, the two confidence intervals begin to agree.  

\begin{table}[h]
\caption{Performance of 95\% confidence intervals versus Monte Carlo sample size ($n$).}
\label{t:rasch}
\centering
\begin{tabular}{ c | c c c || c c c} 
\multicolumn{1}{c}{ } & \multicolumn{3}{c}{coverage prob (\%)} & \multicolumn{3}{c}{median length} \\
$n$ & $10$ & $50$ & $100$ & $10$ & $50$ & $100$ \\ \hline
$\widetilde p$ & 27.1 & 77.9 & 87.9 & 0.36 & 1.00 & 1.10 \\
$\widetilde p_*$ & 99.3 & 98.0 & 96.9 & 2.10 & 1.52 & 1.38 
\end{tabular}
\end{table}

Incidentally, importance sampling seems to be the only currently available method for quickly creating (even approximate) frequentist confidence intervals using conditional inference in this model for a dataset of this size.  Exact conditional logistic regression methods (e.g.,~\citet{stata} or \citet{logxact}) are unable to enumerate the state space or generate direct samples with a few gigabytes of memory.  And the corresponding MCMC methods (e.g.,~\citet{zamar2007elrm}) often fail to adequately sample the state space, even after many hours of computing.  Classical, unconditional, asymptotic confidence intervals also behave poorly.  For the simulation described in Table \ref{t:rasch}, 95\% Wald intervals only had 82.2\% coverage probability with a median length of 1.24.

\subsection{Accelerated testing for neurophysiological data} \label{s:jitter}

The example in this section is similar to the example in Section \ref{s:ex1}.  It is motivated by a practical multiple testing problem that arises in neurophysiology.  Neurophysiologists can routinely record the simultaneous electrical activity of over $100$ individual neurons in a behaving animal's brain.  Each neuron's spiking electrical activity can be modeled as a temporal point process.  Neurophysiologists are interested in detecting various types of spatio-temporal correlations between these point processes, an endeavor which often leads to challenging multiple testing problems.  For example, simply testing for non-zero correlation between all pairs of $100$ neurons gives $\binom{100}{2}$ different tests, and neurophysiologists are often interested in much more complicated situations than this.  This section demonstrates how to use importance sampling to accelerate a Monte Carlo hypothesis test for lagged correlation used in \citet{fujisawa2008behavior}, hereafter FAHB, and then to correctly control for multiple tests via the p-value correction.    

Let $X \eqd (\myvec{T},\myvec{N})$ be a marked point process where $\myvec{T} \eqd (T_1,\dotsc,T_M)$ are the nondecreasing event times taking values in $\{0,1,\dotsc,B-1\}$ and $\myvec{N} \eqd (N_1,\dotsc,N_M)$ are marks taking values in $\{1,\dotsc,C\}$.  Let $\myvec{T}^i \eqd (T^i_1,\dotsc,T^i_{M_i}) \eqd (T_k:N_k=i)$ be the sequence of event times with mark $i$, which in this context are the firing times (rounded to the nearest millisecond) of neuron $i$.  A neuron can fire at most once per ms, so the times in each $\myvec{T}^i$ are strictly increasing.  The dataset used here is described in FAHB and has $C=117$ neurons and was recorded over a period of $B = 2.812\times 10^6$ ms (about $47$ minutes) from a rat performing a working memory task in a simple maze.  There are $M=763501$ total events.  Fix $\Delta=10$ ms.  For each ordered pair of neurons $(i,j)$, $i\neq j$, FAHB were interested in testing the null hypothesis that the joint conditional distribution of $\myvec{T}^i$ and $\myvec{T}^j$ was uniform given $\lfloor{\myvec{T}^i/\Delta}\rfloor$ and $\lfloor{\myvec{T}^j/\Delta}\rfloor$, where $\lfloor{a}\rfloor$ is the greatest integer $\leq a$.  Note that $\lfloor{\myvec{T}^i/\Delta}\rfloor \eqd (\lfloor{T^i_1/\Delta}\rfloor,\dotsc,\lfloor{T^i_{M_i}/\Delta}\rfloor)$ gives a temporally coarsened version of neuron $i$'s firing times --- it only preserves the total number of firings in each length-$\Delta$ window of a fixed partition of time.  FAHB used two different test statistics $t^+(\myvec{T}^i,\myvec{T}^j)$ and $t^-(\myvec{T}^i,\myvec{T}^j)$ defined by
\[ \begin{aligned} t^+(\myvec{T}^i,\myvec{T}^j) & \eqd  \max_{d=1,\dotsc,4} \sum_{k=1}^{M_i}\sum_{\ell=1}^{M_j}\ind\{T_\ell^j-T_k^i=d\} \\
t^-(\myvec{T}^i,\myvec{T}^j) & \eqd  -\min_{d=1,\dotsc,4} \sum_{k=1}^{M_i}\sum_{\ell=1}^{M_j}\ind\{T_\ell^j-T_k^i=d\} \end{aligned} \]
each with power to detect different alternatives that they hoped to distinguish.  (There are some minor technical differences between the tests described here and those reported in FAHB.)  Loosely speaking, the tests were designed to detect a particular type of transient, fast-temporal, lagged-correlation between neurons (via the test statistic) while ignoring all types of slow-temporal correlations (via the conditioning).  Additional motivation, references and details can be found in FAHB.

Testing a specific null and test statistic combination is easy: it is trivial to sample from the null (i.e., conditionally uniform) distribution and to Monte Carlo approximate p-values for any test statistic.  Unfortunately, in this case there are $2C(C-1)=27144$ different hypothesis tests, which leads to a challenging multiple testing problem.  Much like the example in Section \ref{s:ex1}, impractically large Monte Carlo sample sizes are needed to successfully use direct sampling (i.e., $Q\equiv P$) in conjunction with a multiple testing correction such as Bonferroni.  FAHB did not directly address this particular multiple testing problem, but they did use overly conservative individual p-values in hopes of reducing the severity of the problem.  

Here we observe that importance sampling in conjunction with $\widehat p_*$ permits proper and practical use of multiple testing corrections.  The importance sampling proposal distributions that we use are mixture distributions much like the one in Section \ref{s:rasch} (see the Appendix for details).  We use a different proposal distribution for each hypothesis test depending on the choice of test statistic and on the neuron pair.  Using $n=100$ and rejecting whenever $\widehat p_* \leq 0.05/(2C(C-1))$ controls the family-wise error rate (FWER) at level $0.05$.

FAHB reported 78 rejections.  We can confirm about one third of them using this much more stringent FWER procedure.  The qualitative results of FAHB remain the same using this smaller set of rejections, reassuring us that the scientific conclusions of FAHB are not the result of statistical artifacts caused by improper control of multiple hypothesis tests.

\section{Discussion} \label{s:disc}

The practical benefits of using either $\widehat p_*$ or $\widetilde p_*$ are clear.  Valid p-values  are always crucial for multiple testing adjustments.  But even for individual tests, the p-value corrections protect against false rejections resulting from high (and perhaps undiagnosed) variability in the importance sampling approximations.  There is almost no computational penalty for using the corrections.  And there is little or no loss of power for well-behaved importance sampling algorithms.  These advantages extend to confidence intervals constructed by inverting the corrected p-values.  

Monte Carlo approximations should always be accompanied by further approximations of the Monte Carlo standard error.  Approximate standard errors give useful diagnostic information, but they can be very misleading.  The poor performance of $\widehat q$ in Section \ref{s:ex1} shows that the uncorrected p-value approximations cannot be meaningfully corrected by relying on approximate standard errors.  An interesting issue is whether or not the p-value corrections should also be accompanied by approximate standard errors.  Further research is warranted on this issue, but at this time, it seems sensible to report both the uncorrected p-value approximations with their approximate standard errors {\em and} the corrected p-values with no standard errors.  This provides useful information about estimating the target p-value and about Monte Carlo variability, but also permits interpretable hypothesis testing and multiple testing adjustments, without confusing the two issues.  One danger that must be avoided is the temptation to use close agreement between the corrected and uncorrected p-values as an indicator of convergence.  The Appendix contains additional examples that demonstrate how all four p-value approximations (and their standard errors) can be in close agreement but still far from the truth (toward which they are converging).

Although the p-value corrections are extremely simple, they are not always available using off-the-shelf importance sampling software, because the software may not permit evaluation of the importance weight of the original data.  Adding this capability is almost always trivial for the software designer --- instead of actually calling random number generators, simply see what values would lead to the original data and evaluate their probability --- but investigators may be unwilling or unable to modify existing software.  Hopefully, software designers will begin to include this capability in all importance sampling algorithms.  We have found that the generic ability to evaluate a proposal distribution at any point (including the observed data) is quite useful for diagnostic purposes, even if one does not expect to make use of the p-value corrections.

The techniques described here are examples of a more general principle of combining observations from both target and proposal distributions in order to improve Monte Carlo importance sampling approximations.  Suppose, for example, that our computational budget would allow a small sample from the target and a large sample from the proposal.  What is the best way to combine these samples for various types of inference?  As we have shown, even a single observation from the target can be used to provide much more than diagnostic information.  It can be used, surprisingly, to ensure that (often extremely poor) hypothesis test and confidence interval approximations maintain the nominal significance levels.  Perhaps a single observation from the target could also be used to improve approximations of point estimators in some way?

As importance sampling continues to gain prominence as a tool for Monte Carlo approximation, so do the innovations in importance sampling techniques.  Most of these will not fit into the classical importance sampling framework needed for the theorems here.  It remains future work to investigate the degree to which the original data can inform and correct more sophisticated sequential Monte Carlo methods.  But it seems clear that this valuable source of diagnostic information should not be completely ignored.

\appendix


\section{Proofs}

The proofs of both Theorem \ref{t:hat} and Theorem \ref{t:tilde} rely on a simple algebraic fact, encapsulated in the next lemma.
\begin{lem} \label{l} For all $t_0,\dotsc,t_n \in [-\infty,\infty]$ and all $\alpha, w_0,\dotsc,w_n \in [0,\infty]$, we have
\[ \sum_{k=0}^n w_k \ind\left\{\sum_{i=0}^n w_i \ind\{t_i \geq t_k\} \leq \alpha \right\} \leq \alpha \]
\end{lem} 
\begin{proof}
Let $H$ denote the left side of the desired inequality.  We can assume that $H > 0$, since the statement is trivial otherwise.  The pairs $(t_0,w_0),\dotsc,(t_n,w_n)$ can be reordered without affecting the value of $H$, so we can assume that $t_0 \geq \dotsb \geq t_n$. This implies that $\sum_{i=0}^n w_i\ind\{t_i\geq t_k\}$ is increasing in $k$, and that there exists a $k^*$ defined as the largest $k$ for which
\[\sum_{i=0}^n w_i \ind\{t_i \geq t_k\} \leq \alpha \]
So
\[ H = \sum_{k=0}^{k^*} w_k = \sum_{i=0}^{k^*} w_i\ind\{t_i \geq t_{k^*}\} \leq \alpha \qedhere \] 
\end{proof}

\subsection{Proof of Theorem \ref{t:hat}}

For any sequence $\myvec{y} \eqd (y_0,\dotsc,y_n)$ and any $k\in\{0,\dotsc,n\}$, let \[ \myvec{y^k} \eqd (y_k,y_1,\dotsc,y_{k-1},y_0,y_{k+1},\dotsc,y_n) \] which is the sequence obtained by swapping the $0$th element and the $k$th element in the original sequence $\myvec{y}$.  We begin with a lemma. 

\begin{lem} \label{l:hat} For any nonnegative, measurable function $f$, 
\[ \Exp\bigl(f(\myvec{Z})\bigr) = \Exp\left(\frac{1}{n+1}\sum_{k=0}^n w(Y_k)f(\myvec{Y^k})\right) \] 
\end{lem}
\begin{proof}
Recall that $Z_i=Y_i\sim Q$ for $i\neq 0$ and that $Z_0=X\sim P$.  A change of variables from $X$ to $Y_0$ gives
\[ \Exp\bigl(f(\myvec{Z})\bigr) = \Exp\bigl(w(Y_0)f(\myvec{Y})\bigr) \]
Since the distribution of $\myvec{Y}$ is invariant to permutations, we have
\[ \Exp\bigl(w(Y_0)f(\myvec{Y})\bigr) = \Exp\bigl(w(Y_k)f(\myvec{Y^k})\bigr) \]
for each $k=0,\dotsc,n$, which means that
\[ \Exp\bigl(f(\myvec{Z})\bigr) = \Exp\bigl(w(Y_0)f(\myvec{Y})\bigr) = \frac{1}{n+1}\sum_{k=0}^n \Exp\bigl(w(Y_k)f(\myvec{Y^k})\bigr) \]
Moving the sum inside the expectation completes the proof.
\end{proof}

Applying Lemma \ref{l:hat} to the function $f(\myvec{z}) = \ind\bigl\{\widehat p_*(\myvec{z}) \leq \alpha\bigr\}$ gives
\begin{align*}
& \Prob\bigl(\widehat p_*(\myvec{Z})\leq\alpha\bigr) = \Exp\bigl(\ind\bigl\{\widehat p_*(\myvec{Z}) \leq \alpha\bigr\}\bigr) = \Exp\left(\frac{1}{n+1}\sum_{k=0}^n w(Y_k)\ind\bigl\{\widehat p_*(\myvec{Y^k}) \leq \alpha\bigr\} \right) 
\\ & \quad = \Exp\left(\sum_{k=0}^n \frac{w(Y_k)}{n+1}\ind\left\{\sum_{i=0}^n \frac{w(Y_i)}{n+1}\ind\bigl\{t(Y_i,\myvec{Y}) \geq t(Y_k,\myvec{Y})\bigr\}  \leq \alpha\right\} \right) 
\end{align*}
The quantity inside the final expectation is always $\leq \alpha$, which follows from Lemma \ref{l} by taking $t_\ell  \eqd  t(Y_\ell,\myvec{Y})$ and $w_\ell  \eqd  w(Y_\ell)/(n+1)$ for each $\ell=0,\dotsc,n$.  This completes the proof of Theorem \ref{t:hat}.

\subsection{Proof of Theorem \ref{t:tilde}}

Let $\myvec{\Pi}$ denote a random permutation chosen uniformly from $\mathcal{M}$ and chosen independently of $\myvec{Z}$.  Let $\myvec{U} \eqd \myvec{Z^{(\myvec{\pi})}}$.  If ${\myvec{\pi}}\in\mathcal{M}$ is a fixed permutation, then $\myvec{\Pi}$ and $\myvec{\Pi^{(\myvec{\pi})}}$ have the same distribution, which means $\myvec{U}$ and $\myvec{U^{(\myvec{\pi})}}$ have the same distribution.  In particular, $\myvec{U}$ and $\myvec{U^k}$ have the same distribution, where the notation $\myvec{U^k}$ is defined in the proof of Theorem \ref{t:hat}.  We begin with two lemmas, and then the proof of Theorem \ref{t:tilde} is nearly identical to that of Theorem \ref{t:hat}.  We use the convention $0/0 \eqd 0$.

\begin{lem} \label{l:tilde0} Let $P_{\myvec{Z}}$ and $P_{\myvec{U}}$ denote the distributions of $\myvec{Z}$ and $\myvec{U}$, respectively, over $S^{n+1}$.  Then $P_{\myvec{Z}} \ll P_{\myvec{U}}$ and
\[ \frac{dP_{\myvec{Z}}}{dP_{\myvec{U}}}(\myvec{u}) = \frac{(n+1)w(u_0)}{\sum_{j=0}^n w(u_j)} \]
almost surely.
\end{lem}
\begin{proof}
Let $g$ be any nonnegative, measurable function on $S^{n+1}$.  We need only show that
\begin{equation} \label{e:l:tilde0} \Exp\bigl(g(\myvec{Z})\bigr) = \Exp\left(\frac{(n+1)w(U_0)}{\sum_{j=0}^n w(U_j)}g(\myvec{U})\right) \end{equation}
For any $\myvec{\pi}\in\mathcal{M}$ there is a unique inverse permutation $\myvec{\pi^{-1}}\in\mathcal{M}$ with $\pi^{-1}_{\pi_i}=\pi_{\pi^{-1}_i}=i$, for each $i=0,\dotsc,n$.  Comparing with the proof of Lemma \ref{l:hat}, for any nonnegative, measurable function $f$ we have
\begin{equation} \label{e:tilde0_1}
\Exp\bigl(f(\myvec{Z^{(\myvec{\pi})}})\bigr) = \Exp\bigl(f(\myvec{Y^{(\myvec{\pi})}})w(Y_0)\bigr) = \Exp\bigl(f(\myvec{Y})w(Y_{\pi^{-1}_0})\bigr)
\end{equation}
where the first equality is a change of variables from $\myvec{Z}$ to $\myvec{Y}$ and the second equality follows from the fact that the distribution of $\myvec{Y}$ is permutation invariant, so, in particular, $\myvec{Y}$ and $\myvec{Y^{(\myvec{\pi^{-1}})}}$ have the same distribution.  Using \eqref{e:tilde0_1} gives
\begin{align}
\notag & \Exp\bigl(f(\myvec{U})\bigr) = \frac{1}{(n+1)!}\sum_{\myvec{\pi}\in\mathcal{M}}\Exp\bigl(f(\myvec{U})\bigl|\myvec{\Pi}=\myvec{\pi}\bigr) = \frac{1}{(n+1)!}\sum_{\myvec{\pi}\in\mathcal{M}}\Exp\bigl(f(\myvec{Z^{(\myvec{\pi})}})\bigr)
\\ \notag & \quad = \frac{1}{(n+1)!}\sum_{\myvec{\pi}\in\mathcal{M}}\Exp\bigl(f(\myvec{Y})w(Y_{\pi^{-1}_0})\bigr) = \frac{1}{(n+1)!}\sum_{j=0}^n \sum_{\substack{\myvec{\pi}\in\mathcal{M} :\\\pi^{-1}_0=j}}\Exp\bigl(f(\myvec{Y})w(Y_j)\bigr)
\\ & \label{e:tilde0_2} \quad = \frac{1}{(n+1)!}\sum_{j=0}^n n! \Exp\bigl(f(\myvec{Y})w(Y_j)\bigr) = \Exp\left(\frac{\sum_{j=0}^n w(Y_j)}{n+1}f(\myvec{Y}) \right) 
\end{align}
Applying \eqref{e:tilde0_2} to the function $f(\myvec{u})=(n+1)g(\myvec{u})w(u_0)/\bigl(\sum_{j=0}^n w(u_j)\bigr)$ gives
\begin{align}
\notag & \Exp\left(\frac{(n+1)w(U_0)}{\sum_{j=0}^n w(U_j)}g(\myvec{U})\right) =  \Exp\left(\frac{\sum_{j=0}^n w(Y_j)}{n+1}\frac{(n+1)w(Y_0)}{\sum_{j=0}^n w(Y_j)}g(\myvec{Y})\right)
\\ \notag & \quad = \Exp\bigl(w(Y_0)g(\myvec{Y})\bigr) = \Exp\bigl(g(\myvec{Z})\bigr)
\end{align}
where the last equality is a change of variables as in \eqref{e:tilde0_1}.  This gives \eqref{e:l:tilde0} and completes the proof.
\end{proof}

\begin{lem} \label{l:tilde} For any nonnegative, measurable function $f$,
\[ \Exp\bigl(f(\myvec{Z})\bigr) = \Exp\left(\sum_{k=0}^n \frac{w(U_k)}{\sum_{j=0}^n w(U_j)}f(\myvec{U^k})\right) \] 
\end{lem}
\begin{proof}
Changing variables from $\myvec{Z}$ to $\myvec{U}$ and using Lemma \ref{l:tilde0} gives 
\[ \Exp\bigl(f(\myvec{Z})\bigr) = \Exp\left(\frac{(n+1)w(U_0)}{\sum_{j=0}^n w(U_j)} f(\myvec{U})\right) \]
Since the distribution of $\myvec{U}$ is invariant to permutations, we have
\[ \Exp\left(\frac{(n+1)w(U_0)}{\sum_{j=0}^n w(U_j)} f(\myvec{U})\right) = \Exp\left(\frac{(n+1)w(U_k)}{\sum_{j=0}^n w(U_j)} f(\myvec{U^k})\right) \]
for each $k=0,\dotsc,n$, which means that
\[ \Exp\bigl(f(\myvec{Z})\bigr) = \Exp\left(\frac{(n+1)w(U_0)}{\sum_{j=0}^n w(U_j)} f(\myvec{U})\right) = \frac{1}{n+1}\sum_{k=0}^n \Exp\left(\frac{(n+1)w(U_k)}{\sum_{j=0}^n w(U_j)} f(\myvec{U^k})\right) \]
Moving the sum inside the expectation and cancelling the $(n+1)$'s completes the proof.
\end{proof}

Applying Lemma \ref{l:tilde} to the function $f(\myvec{z}) = \ind\bigl\{\widetilde p_*(\myvec{z}) \leq \alpha\bigr\}$ gives
\begin{align*}
& \Prob\bigl(\widetilde p_*(\myvec{Z})\leq\alpha\bigr) = \Exp\bigl(\ind\bigl\{\widetilde p_*(\myvec{Z}) \leq \alpha\bigr\}\bigr) = \Exp\left(\sum_{k=0}^n \frac{w(U_k)}{\sum_{j=0}^n w(U_j)} \ind\bigl\{\widetilde p_*(\myvec{U^k}) \leq \alpha\bigr\} \right) 
\\ & \quad = \Exp\left(\sum_{k=0}^n \frac{w(U_k)}{\sum_{j=0}^n w(U_j)}\ind\left\{\sum_{i=0}^n \frac{w(U_i)}{\sum_{j=0}^n w(U_j)}\ind\bigl\{t(U_i,\myvec{U}) \geq t(U_k,\myvec{U})\bigr\}  \leq \alpha\right\} \right) 
\end{align*}
The quantity inside the final expectation is always $\leq \alpha$, which follows from Lemma \ref{l} by taking $t_\ell  \eqd  t(U_\ell,\myvec{U})$ and $w_\ell  \eqd  w(U_\ell)/\bigl(\sum_{j=0}^n w(U_j)\bigr)$ for each $\ell=0,\dotsc,n$.  This completes the proof of Theorem \ref{t:tilde}.

\section{Proposal distributions and simulation details}

The simulation example in Section \ref{s:ex1} uses conditional inference, so the target and proposal distributions for each dataset $i$ (notation suppressed) can depend on the observed values $\myvec{V} \eqd (V_1,\dotsc,V_m)$ and the fact that there are $r$ one-labels and $m-r$ zero-labels (but cannot depend on the observed pairing of labels and values). 
Let $\myvec{I} \eqd (I_1,\dotsc,I_m)$ be a permutation that makes $\myvec{V}$ non-increasing, i.e., $V_{I_1}\geq\dotsb\geq V_{I_m}$.  Choose a random permutation $\myvec{\Pi} \eqd (\Pi_1,\dotsc,\Pi_m)$ according to the distribution  
\[ \Prob(\myvec{\Pi}=\myvec{\pi})  \eqd  \frac{\exp\bigr(\theta\sum_{i=1}^r \ind\{\pi_i \leq r\}\bigl)}{r!(m-r)!\sum_{k=0}^m \binom{r}{k}\binom{m-r}{r-k} \exp(\theta k)} \quad (\theta\in\mathbb{R}, r=0,\dotsc,m) \]  
where the binomial coefficients $\binom{a}{b} \eqd a!/(b!(a-b)!)$ are defined to be zero if $a < 0$, $b < 0$, or $a<b$.  Leaving $\myvec{V}$ in the original observed order and permuting $\myvec{L}$ so that $L_{I_{\Pi_1}}=\dotsb=L_{I_{\Pi_r}}=1$ and $L_{I_{\Pi_{r+1}}}=\dotsb=L_{I_{\Pi_m}}=0$ gives a random pairing of values with labels.  The case $\theta=0$ is the uniform distribution over permutations, i.e., the target null conditional distribution.  We used $\theta=3$ for the proposal distribution which assigns higher probability to those permutations that tend to match the label one with larger values.

The simulation example in Section \ref{s:jitter} also uses conditional inference, so the target and proposal distributions for each pair of neurons $(i,j)$ can depend on the temporally coarsened event times $\lfloor{\myvec{T}^i/\Delta}\rfloor$ and $\lfloor{\myvec{T}^j/\Delta}\rfloor$.  Define the set of all event times with the same temporal coarsening as neuron $i$ to be
\[ \Omega_i  \eqd  \bigl\{\myvec{u}\in\{0,\dotsc,B-1\}^{M_i}:u_1<\dotsb<u_{M_i},\ \lfloor{u_k/\Delta}\rfloor=\lfloor{T_k^i/\Delta}\rfloor, \forall k\bigr\} \]
For any $\myvec{s}\in\mathbb{Z}^m$ and $a\in\mathbb{Z}$ define
\[ R_a(\myvec{s})  \eqd  \bigl|\{k:\lfloor{s_k/\Delta}\rfloor=a\}\bigr| \] so that, for example,
\[ \bigl|\Omega_i\bigr| = \prod_a \binom{\Delta}{R_a(\myvec{T}^i)} \]
The target distribution (i.e., the null hypothesis for pair $(i,j)$) is the uniform distribution over $\Omega_i\times\Omega_j$.  For any $\myvec{s}\in\Omega_j$, $\myvec{u}\in\Omega_i$, $d\in\mathbb{Z}$, and $\theta\in\mathbb{R}$ define
\[ \rho(\myvec{s},\myvec{u},d,\theta)  \eqd  \frac{\exp\bigl(\theta\sum_{k=1}^{M_j}\sum_{\ell=1}^{M_i}\ind\{s_k=u_\ell+d\}\bigr)}{\prod_a \sum_{r=0}^\Delta \binom{R_a(\myvec{u}+d)}{r}\binom{\Delta-R_a(\myvec{u}+d)}{R_a(\myvec{s})-r}\exp(\theta r) }\ind\{\myvec{s}\in\Omega_j\} \]
which is a probability distribution over $\Omega_j$ for fixed $\myvec{u}$, $d$, and $\theta$.  The case $\theta=0$ is the uniform distribution over $\Omega_j$; the case $\theta > 0$ prefers those $\myvec{s}$ with event times that match those in $\myvec{u}+d$; and the case $\theta < 0$ prefers those $\myvec{s}$ that do not.  We used proposal distributions of the form
\[ \Prob\bigl(\myvec{T}^i=\myvec{u},\myvec{T}^j=\myvec{s}\bigr)  \eqd  \frac{1}{5}\sum_{d=0}^4 \frac{\ind\{\myvec{u}\in\Omega_i\}}{\bigl|\Omega_i\bigr|}\rho(\myvec{s},\myvec{u},d,\theta_d) \]
which choose event times uniformly for neuron $i$ and then choose times in neuron $j$ that prefer (or avoid) times of a specific lag from those of neuron $i$.  For the test statistic $t^+$ we used $(\theta_0,\dotsc,\theta_4) \eqd (0,.5,.5,.5,.5)$ and for $t^-$ we used $(0,-.5,-.5,-.5,-.5)$.


All simulations and analysis were performed with custom software written in Matlab 2010b and executed on a 2.66 GHz iMac with 8GB of RAM using Matlab's default pseudo random number generator (Mersenne Twister).  The parameters used for the logistic regression example in Section \ref{s:rasch} are $\kappa=-1.628$, $\myvec{\beta}=($0.210, -0.066,  0.576, -0.197,  0.231,  0.184, -0.034, -0.279, -0.396,  0.000$)$ and $\myvec{\alpha}=($-0.183, -0.735, -0.144, -0.756, -0.749, -0.226, -0.538, -0.213, -0.118, -0.284, -0.127, -0.632, -0.132,  0.104,  0.000, -0.781, -0.500, -0.498, -0.182, -0.269, -0.077, -0.499, -0.661, -0.780, -0.095, -0.661, -0.478, -0.315, -0.638, -0.225, -0.382, -0.715, -0.085, -0.766, -0.573, -0.629, -0.336, -0.775, -0.461, -0.762, -0.754, -0.082, -0.575, -0.263,  0.098, -0.434, -0.172, -0.109, -0.434, -0.211, -0.757,  0.067, -0.679, -0.601, -0.069, -0.379, -0.098, -0.471, -0.594, -0.830, -0.193, -0.437, -0.415, -0.257, -0.807, -0.551, -0.094, -0.170, -0.741, -0.737, -0.774, -0.859, -0.444, -0.211, -0.144, -0.336, -0.758, -0.235, -0.740, -0.732, -0.768, -0.725, -0.698, -0.671, -0.549, -0.550, -0.649, -0.616,  0.026, -0.164, -0.311, -0.682, -0.655, -0.789,  0.047, -0.160, -0.309, -0.553, -0.701, -0.244,  0.121, -0.696, -0.609, -0.470, -0.793, -0.183, -0.464,  0.116, -0.465, -0.246, -0.712, -0.485, -0.706, -0.109,  0.004, -0.516, -0.181, -0.573, -0.336, -0.034, -0.269, -0.531, -0.568, -0.414, -0.444, -0.507, -0.308, -0.124, -0.442, -0.437, -0.742, -0.842, -0.577, -0.549, -0.213,  0.090,  0.069, -0.409, -0.626, -0.103, -0.107, -0.126, -0.123, -0.761, -0.185, -0.403, -0.655, -0.768, -0.043, -0.692, -0.703, -0.201,  0.028, -0.350, -0.164, -0.713,  0.087, -0.326, -0.187, -0.830, -0.058, -0.118, -0.747, -0.342, -0.541, -0.320, -0.468, -0.452, -0.686, -0.611, -0.846,  0.057, -0.213,  0.066, -0.703,  0.054, -0.072, -0.289, -0.427, -0.609, -0.115, -0.638, -0.803, -0.099, -0.196, -0.152, -0.225, -0.448, -0.476, -0.051, -0.549, -0.052, -0.078, -0.014, -0.361, -0.231,  0.084, -0.423, -0.807,  0.000$)$.  The final entries of $\myvec{\beta}$ and $\myvec{\alpha}$ were constrained to be zero to avoid identifiability problems.  The covariates $\myvec{\nu}$ are recorded as a $200\times 10$ matrix in the ASCII file {\tt nu.csv}.

\section{Additional examples}

\subsection{Mismatched Gaussians} \label{s:G}

For the first example, we take $P$ to be Normal$(0,1)$, $Q$ to be Normal$(\mu,\sigma)$, and use the test statistic $t(x) \eqd x$.  We experiment with a few choices of $\mu$, $\sigma$ and $n$, and in each case compute the cumulative distribution functions (cdfs) and mean-squared errors (MSEs) of the approximations $\widehat p$, $\widetilde p$, $\widehat p_*$, and $\widetilde p_*$ of the true p-value $p(X) = 1-\Phi(X)$, where $\Phi$ is the standard normal cdf.  
Actually, we do not compute these quantities directly, but approximate them using $10^6$ Monte Carlo samples (each of which uses a new $X$ and new importance samples $Y_1,\dotsc,Y_n$).   
Figure \ref{f:1} and Table \ref{t:1} show the results.  
In each case, $\widehat p_*$ and $\widetilde p_*$ seem to be valid p-values, confirming the theory.  There are many cases where $\widehat p$ and $\widetilde p$ are not valid.  

\begin{figure}[h] 
\centering
	\epsfig{file=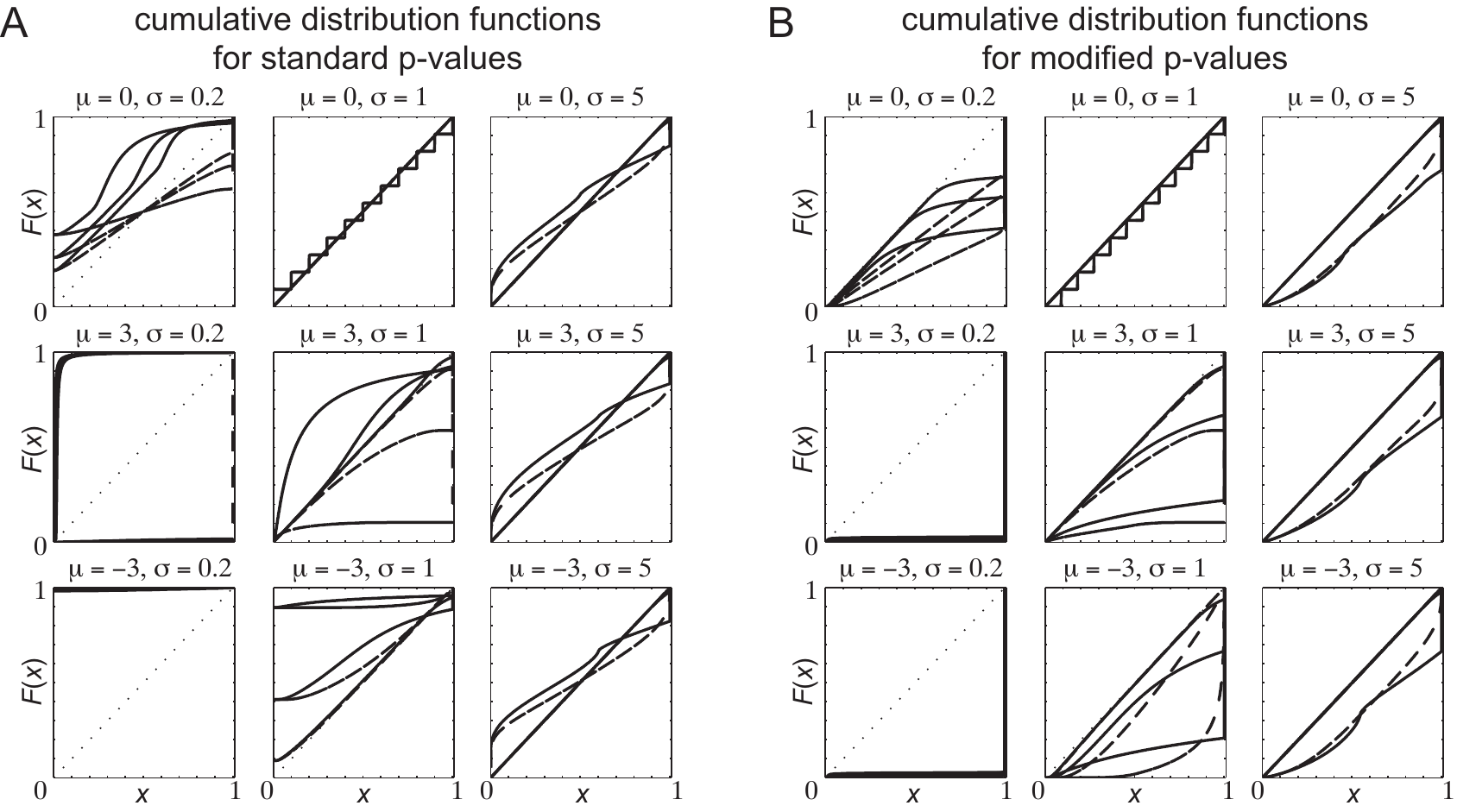,width=\textwidth}
\caption{Estimated cumulative distribution functions (cdfs) for the approximate p-values in Example \ref{s:G} under the null hypothesis.  Each plot corresponds to a specific choice of $\mu$ and $\sigma$ for the proposal distribution.  The plots in panel A show the cdfs for $\widehat p$ (solid lines) and $\widetilde p$ (dashed lines) for each of $n=10, 1000, 100000$, for a total of six cdfs on each graph.  The plots in panel B show $\widehat p_*$ (solid lines) and $\widetilde p_*$ (dashed lines) for each of $n=10, 1000, 100000$.  (Each cdf was estimated using $10^6$ Monte Carlo experiments, each of which involved a new $X$ and new $Y_1,\dotsc,Y_n$.  Since cdfs 
take values in $[0,1]$, the maximum possible variance of any of these estimators is $(1/4)10^{-6}$, which implies that an interval of $\pm 10^{-3}$ around any estimate is at least a 95\% confidence interval.)  The cdf of the target p-value, $p$, is the cdf of a Uniform$(0,1)$ random variable, namely, $F(x)=x$, and is shown with a dotted line along the diagonal of each plot.  The different choices of $n$ are not labeled on the plots, but where the curves are distinguishable, larger $n$ will always be closer to the diagonal, since the approximate p-values always converge to the target p-value.  A {\em valid} p-value has a cdf that lies on or below the diagonal.  If the cdf drops strictly below the diagonal, this indicates a loss of power.  If the cdf exceeds the diagonal, this indicates that the resulting hypothesis test does not correctly control the type I error.  Note that the corrected  p-values are always valid, but that the original p-values are often invalid. \label{f:1}} 
\end{figure}

\begin{table}[h]
\caption{Estimated MSEs of the approximate p-values (smallest for each row is bold)}
\label{t:1}
\centering
\begin{tabular}{c c | c || c c c c}
\multicolumn{3}{l}{\hspace{4ex}$Q$} & \multicolumn{4}{c}{MSE} \\ 
$\mu$ & $\sigma$ & $n$ & $\widehat p$ & $\widehat p_*$ & $\widetilde p$ & $\widetilde p_*$ \\
[0.5ex] \hline \hline
$0$ & $0.2$ & $10$ & $1.16e-01$ & $2.54e-01$ & $\bs{5.12e-02}$ & $2.37e-01$ \\
$0$ & $0.2$ & $1000$ & $5.98e-02$ & $1.98e-01$ & $\bs{2.43e-02}$ & $1.75e-01$ \\
$0$ & $0.2$ & $100000$ & $3.39e-02$ & $1.52e-01$ & $\bs{1.41e-02}$ & $1.33e-01$ \\ \hline
$0$ & $1$ & $10$ & $1.67e-02$ & $\bs{1.66e-02}$ & $1.67e-02$ & $\bs{1.66e-02}$ \\
$0$ & $1$ & $1000$ & $\bs{1.67e-04}$ & $1.67e-04$ & $\bs{1.67e-04}$ & $1.67e-04$ \\
$0$ & $1$ & $100000$ & $\bs{1.66e-06}$ & $1.66e-06$ & $\bs{1.66e-06}$ & $1.66e-06$ \\ \hline
$0$ & $5$ & $10$ & $8.58e-02$ & $9.52e-02$ & $6.80e-02$ & $\bs{6.37e-02}$ \\
$0$ & $5$ & $1000$ & $1.38e-03$ & $1.39e-03$ & $\bs{4.96e-04}$ & $4.96e-04$ \\
$0$ & $5$ & $100000$ & $1.45e-05$ & $1.45e-05$ & $\bs{4.94e-06}$ & $4.94e-06$ \\ \hline
$3$ & $0.2$ & $10$ & $\bs{3.23e-01}$ & $3.23e-01$ & $3.31e-01$ & $3.31e-01$ \\
$3$ & $0.2$ & $1000$ & $3.17e-01$ & $\bs{3.15e-01}$ & $3.26e-01$ & $3.26e-01$ \\
$3$ & $0.2$ & $100000$ & $3.08e-01$ & $\bs{3.07e-01}$ & $3.20e-01$ & $3.20e-01$ \\ \hline
$3$ & $1$ & $10$ & $1.94e-01$ & $\bs{1.83e-01}$ & $2.53e-01$ & $2.56e-01$ \\
$3$ & $1$ & $1000$ & $2.48e-02$ & $\bs{2.04e-02}$ & $5.18e-02$ & $5.06e-02$ \\
$3$ & $1$ & $100000$ & $7.69e-04$ & $\bs{6.06e-04}$ & $4.76e-03$ & $4.65e-03$ \\ \hline
$3$ & $5$ & $10$ & $9.62e-02$ & $1.05e-01$ & $8.63e-02$ & $\bs{7.92e-02}$ \\
$3$ & $5$ & $1000$ & $1.62e-03$ & $1.63e-03$ & $\bs{5.97e-04}$ & $5.98e-04$ \\
$3$ & $5$ & $100000$ & $1.72e-05$ & $1.72e-05$ & $5.93e-06$ & $\bs{5.93e-06}$ \\ \hline
$-3$ & $0.2$ & $10$ & $3.34e-01$ & $3.42e-01$ & $\bs{3.30e-01}$ & $3.34e-01$ \\
$-3$ & $0.2$ & $1000$ & $3.33e-01$ & $3.49e-01$ & $\bs{3.26e-01}$ & $3.36e-01$ \\
$-3$ & $0.2$ & $100000$ & $3.32e-01$ & $3.57e-01$ & $\bs{3.20e-01}$ & $3.39e-01$ \\ \hline
$-3$ & $1$ & $10$ & $2.80e-01$ & $3.99e-01$ & $\bs{2.53e-01}$ & $3.20e-01$ \\
$-3$ & $1$ & $1000$ & $1.05e-01$ & $2.68e-01$ & $\bs{5.17e-02}$ & $1.44e-01$ \\
$-3$ & $1$ & $100000$ & $1.62e-02$ & $3.22e-02$ & $\bs{4.79e-03}$ & $1.21e-02$ \\ \hline
$-3$ & $5$ & $10$ & $1.06e-01$ & $1.23e-01$ & $8.60e-02$ & $\bs{7.82e-02}$ \\
$-3$ & $5$ & $1000$ & $1.79e-03$ & $1.81e-03$ & $\bs{5.98e-04}$ & $5.98e-04$ \\
$-3$ & $5$ & $100000$ & $1.88e-05$ & $1.88e-05$ & $\bs{5.93e-06}$ & $5.93e-06$ \\ \hline
\end{tabular}
\end{table}

The MSEs of all the estimators seem comparable, with none uniformly better than the rest, although for a majority of these examples $\widetilde p$ has the smallest MSE.  (We use $\min\{\widehat p,1\}$ and $\min\{\widehat p_*,1\}$ instead of $\widehat p$ and $\widehat p_*$, respectively, for computing MSE.)  
An important special case among these examples is $\mu=0$, $\sigma=0.2$, for which $\widehat p$ and $\widetilde p$ have MSEs up to 10 times smaller than the MSEs of $\widehat p_*$ and $\widetilde p_*$, respectively.  For this special case, none of the estimators does a good job approximating $p$, even with $n=10^5$, but $\widehat p_*$ and $\widetilde p_*$ are conservative, whereas $\widehat p$ and $\widetilde p$ are strongly liberal (especially for small p-values).  In many hypothesis testing contexts, the conservative choice is desirable, despite worse MSE.

\subsection{Conditional testing in binary tables} \label{s:binary}

The next example was inspired by \citet{Bezakova:Negative:2006} which describes a situation where an importance sampling algorithm taken from the literature converges extremely slowly.  Let $X$ be a binary $52\times 102$ matrix with row sums $(51,1,1,\dotsc,1)$ and column sums $(1,1,1,\dotsc,1)$.  The null hypothesis is that $X$ came from the uniform distribution over the class of all binary matrices with these row and column sums.  The alternative hypothesis is that $X$ came from a non-uniform distribution that prefers configurations where the 51 ones in the first row tend to clump near the later columns.  For any binary matrix $x$ with these row and column sums, let $\ell_1(x) < \dotsb < \ell_{51}(x)$ denote the indices of the columns for which $x$ has ones in the first row.  For example, if the first row is $(0,1,0,1,0,\dotsc,1)$, then $(\ell_1,\dotsc,\ell_{51})=(2,4,6,\dotsc,102)$.  A suitable test statistic for distinguishing the null from the alternative is
\[ t(x)  \eqd  \sum_{j=1}^{51} \ell_j(x) \]
Suppose that we observe a matrix $X$ with $\bigl(\ell_1(X),\dotsc,\ell_{51}(X)\bigr) = ($1,  7,  8, 10, 11, 15, 16, 17, 20, 21, 28, 29, 30, 36, 37, 40, 41, 42, 48, 49, 51, 54, 55, 56, 57, 58, 60, 61, 62, 63, 65, 67, 68, 69, 70, 73, 75, 77, 80, 81, 82, 85, 86, 87, 91, 92, 94, 95, 96, 97, 100$)$, so that $t(X)=2813$.  What is the p-value $p(X)$?

\citet{Chen:Sequential:2005} suggest a proposal distribution for approximate uniform generation of binary matrices with specified row and column sums.
(They actually suggest a variety of proposal distributions.  We use the same algorithm that they used in their examples: sampling columns successively and using conditional Poisson sampling weights of $r_i/(n-r_i)$.  Note that the basic and more delicate algorithms that they describe are identical for the example here.  Note also that we could have used their algorithm to sample rows successively, in which case the sampling would be exactly uniform because of the symmetry for this particular example.)
    Applying their algorithm to this problem with $n=2\times10^7$ samples from the proposal distribution gives approximate p-values and standard errors of $\widetilde p = 0.068\pm 0.013$ and the correction $\widetilde p_* = 0.068$.  (We also have $\widehat p = 0.066\pm0.014$ and $\widehat p_* = 0.066$, although normally they would not be available because the importance weights are known only up to a constant of proportionality.  The exact weights are available in this particular example because of the special symmetry.)
The estimated squared coefficient of variation (i.e., the sample variance divided by the square of the sample mean) for the importance weights is $\widehat{cv}^2=20249$, which is extremely large, but which suggests an effective sample size of around $n/(1+cv^2)\approx 988$ according to the heuristic in \citet{Kong:Sequential:1994}.  While much smaller than $2\times 10^7$, a sample size of $988$ would usually ensure reasonable convergence using direct sampling and the estimated standard errors seem to agree with this heuristic.  Furthermore, the solid black line in Figure \ref{f:2} shows how $\widetilde p$ changes with $n$.  It seems to have stabilized by the end.  Thus, we might be tempted to take $\widetilde p = 0.068$ as a reasonable approximation of $p(X)$ and report a ``p-value'' of $0.068$.

\begin{figure}[h] 
\begin{center}
	\epsfig{file=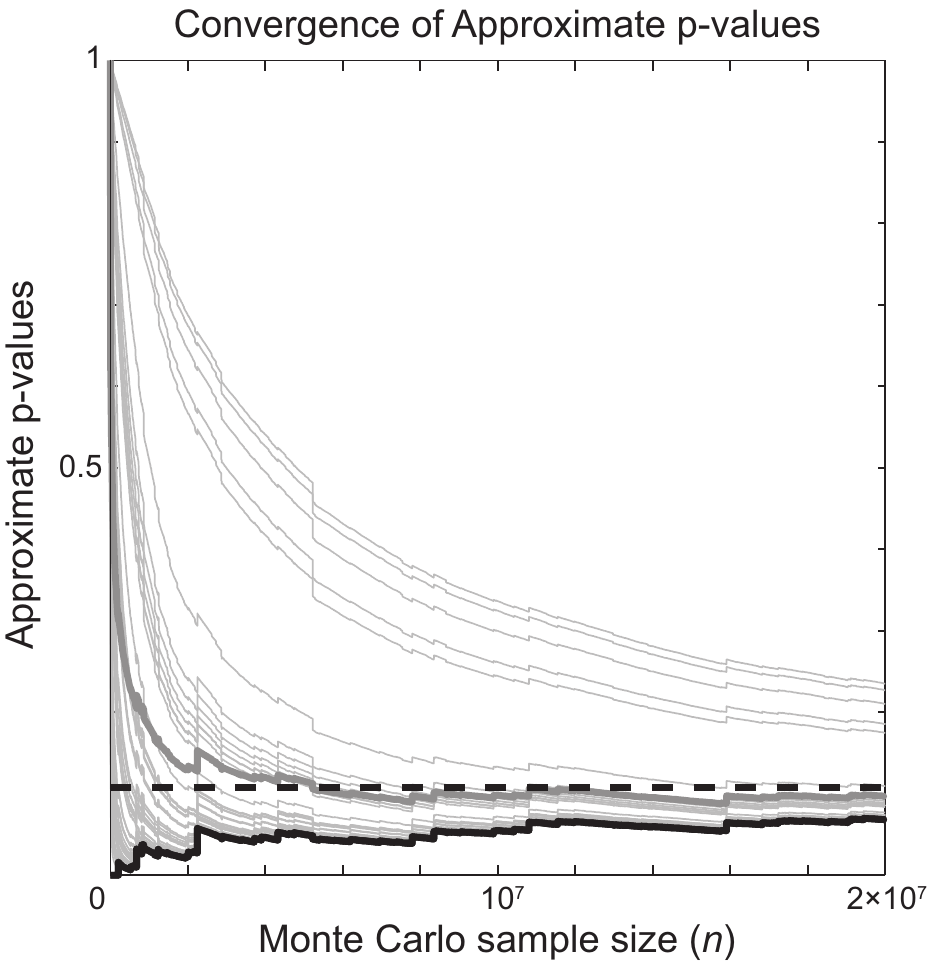,width=.5\textwidth}
\end{center}
\caption{Evolution of the approximated p-values with increasing Monte Carlo sample size ($n$) in Example \ref{s:binary}.  The dashed black line shows (an excellent approximation of) the target p-value.  All of the other lines in the plot will converge to this value as $n\to\infty$, because they are all consistent estimators.  The solid black line shows the uncorrected  importance sampling p-value approximation, $\widetilde p$.  For a fixed value of the observed test statistic, $t(X)$, this line does not depend on the choice of $X$.  The thin gray lines show the corrected  importance sampling p-value approximation, $\widetilde p_*$, for different choices of $X$, but each with the same value of the observed test statistic, $t(X)$.  Unlike $\widetilde p$, $\widetilde p_*$ does depend on the choice of $X$.  The figure shows 50 thin gray lines.  The thick gray line is the average of 1000 such thin gray lines (the first 50 of which are shown).  The $X$'s for these 1000 were chosen uniformly from the set of $X$'s with the same test statistic.  The same sequence of $2\times 10^7$ importance samples were used in all cases. \label{f:2}}    
\end{figure}

It turns out that $0.068$ is not a good approximation of $p(X)$.  The special symmetry in this particular example permits exact and efficient uniform sampling --- each of the $\binom{102}{51}$ choices for the first row is equally likely, and each of the remaining $51!$ choices for the remainder of the matrix is equally likely.  The true p-value is very nearly $p(X)=0.107$, as estimated by $\bar p$ using $n=10^7$ i.i.d.~Monte Carlo samples from $P$.  This is shown as the black dashed line in Figure \ref{f:2}.  

Despite the fact that $\widetilde p$ and $\widetilde p_*$ are essentially identical for this example, and despite the fact that they do not give a good estimate of the target p-value, $p(X)$, the suggestion here is that reporting $\widetilde p_*$ as a p-value correctly preserves the interpretation of a p-value, whereas, reporting $\widetilde p$ as a p-value is misleading.  If we were to repeat this test many times, generating new $X$'s sampled from the null (uniform) distribution and new importance samples, we would find that $\widetilde p$ is much too liberal --- there will be too many small p-values.  The correction in $\widetilde p_*$ does not have this problem.

Generating $2\times 10^7$ samples from $Q$ is computationally demanding, so we do not have a simulation experiment to illustrate the assertion in the previous paragraph.  We do, however, illustrate the idea with two less demanding computations.  The first approximates cdfs of $\widetilde p$ and $\widetilde p_*$ for $n=1000$, much like the experiments in Example \ref{s:G}.  Results are shown in Figure \ref{f:3}.  The uncorrected  p-values are extremely liberal --- rejecting when $\widetilde p \leq 0.05$ results in a type I error-rate of 40\%.  The corrected  p-values are conservative, but valid, correctly preserving the interpretation of a level $\alpha$ test.  

\begin{figure}[h] 
\begin{center}
	\epsfig{file=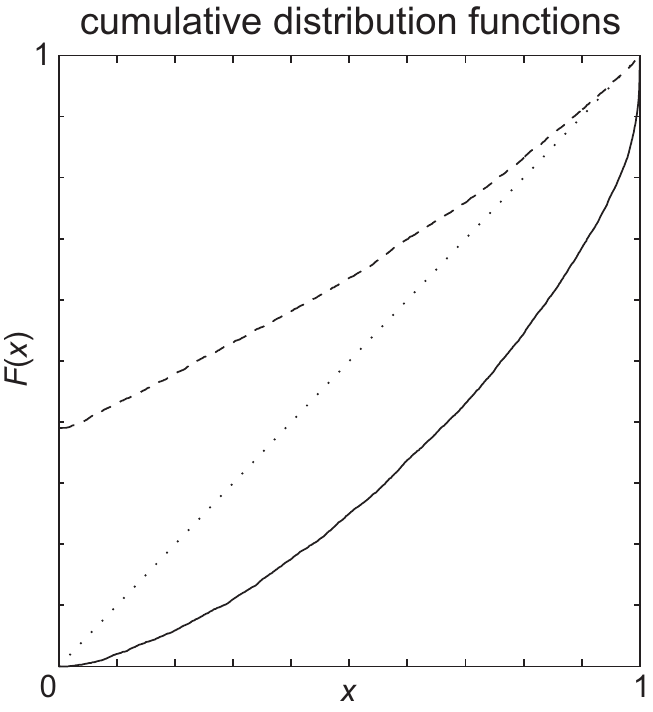,width=.4\textwidth}
\end{center}
\caption{Cumulative distribution functions under the null hypothesis for $\widetilde p$ (dashed line) and $\widetilde p_*$ (solid line) using $n=1000$ in Example \ref{s:binary}.  The cdfs are approximated with $10^4$ Monte Carlo repetitions.  The dotted line is the Uniform$(0,1)$ cdf.  It is visually indistinguishable from the cdf of the target p-value, $p(X)$, although not identical, because of the discreteness of the test statistic. \label{f:3}}    
\end{figure}

Our second illustration uses the same $2\times 10^7$ importance samples as before, but simply changes the original observation $X$ to a new one with the same value of $t$.  For example, consider another observation of $X$ with $\bigl(\ell_1(X),\dotsc,\ell_{51}(X)\bigr) = ($1,  2,  6,  8, 10, 14, 16, 18, 19, 20, 21, 23, 25, 29, 32, 33, 36, 38, 42, 44, 46, 49, 50, 53, 54, 57, 60, 62, 63, 67, 68, 69, 70, 72, 73, 76, 78, 81, 82, 88, 89, 92, 93, 94, 95, 96, 97, 99, 100, 101, 102$)$, which has the same value of the test statistic $t(X)=2813$, and consequently the same target p-value.  It also has the same approximate uncorrected  p-value, $\widetilde p = 0.068$ (using the same importance samples from before).  But it does not have the same corrected  p-value.  In this case $\widetilde p_* = 0.3907$.  The change results from the fact that this new observation of $X$ is much rarer than the previous one under the importance sampling proposal distribution --- so rare, in fact, that even $n=2\times 10^7$ is insufficient to mitigate the effect that its importance weight has on the corrected  p-value.  The 50 different thin gray lines in Figure \ref{f:2} show how $\widetilde p_*$ varies with $n$ using 50 different choices of $X$, all with $t(X)=2813$.  The thick gray line shows the average of 1000 different choices of $X$, all with $t(X)=2813$.  The new $X$'s were chosen uniformly subject to the constraint that $t(X)=2813$, which is actually quite efficient in this special case by using direct sampling from $P$, followed by rejection sampling.  By paying attention to $X$, and not just $t(X)$, $\widetilde p_*$ can be a valid p-value, even though it is not a good estimator of $p(X)$ using $n\approx 10^7$.

This class of importance sampling algorithms is used, for example, in testing goodness of fit in Rasch models and in the statistical analysis of ecological data.  \citet{Chen:Sequential:2005} use an occurrence matrix for ``Darwin's finch data'' to illustrate their importance sampling approach for testing the uniformity of zero-one tables.  The observed table, $X$, is a $13\times 17$ binary matrix indicating which of 13 species of finch inhabit which of 17 islands in the Gal\'apagos.  The (ordered) row sums of $X$ are (17, 14, 14, 13, 12, 11, 10, 10, 10,  6,  2,  2,  1) and the (ordered) column sums are (11, 10, 10, 10, 10,  9,  9,  9,  8,  8,  7,  4,  4,  4,  3,  3,  3).  A scientifically relevant null hypothesis is that $X$ was selected uniformly among all possible occurrence matrices with the same sequence of row and column sums.  A scientifically relevant test statistic is
\[ t(x)  \eqd  \frac{1}{13(12)}\sum_{i\neq j} \bigl((xx^{\text{T}})_{ij}\bigr)^2 \]
where $(A)_{ij}$ denotes entry $(i,j)$ in matrix $A$, and where $A^{\text{T}}$ denotes the transpose of $A$.  The test statistic should be large if there is competition among species.  The observed data has $t(X)=53.1$.  See \citet{Chen:Sequential:2005} for details and references.

\citet{Chen:Sequential:2005} use the importance sampling algorithm described above to generate an approximate p-value for this hypothesis test.  They report p-values of $\widetilde p = (4.0\pm 2.8)\times 10^{-4}$ using $n=10^4$ and $\widetilde p = (3.96\pm 0.36)\times 10^{-4}$ using $n=10^6$.  The error terms are approximate standard errors, estimated from the importance samples.  Repeating their analyses, but with new importance samples, gives $\widetilde p = (7.77\pm7.78)\times 10^{-4}$ and $\widetilde p_* = 13.92\times 10^{-4}$ using $n=10^4$, and $\widetilde p = (4.32\pm0.51)\times 10^{-4}$ and $\widetilde p_* = 4.38\times 10^{-4}$ using $n=10^6$.  An investigator reporting any of the values of $\widetilde p$, even with estimates of standard errors, has few guarantees.  But reporting either of the values of $\widetilde p_*$ is guaranteed to preserve the usual interpretation of a p-value.

\bibliographystyle{apacite}
\bibliography{ISpvalue}

\end{document}